\documentclass[11pt]{article}
\usepackage{amssymb}
\usepackage{mathtools}
\usepackage{graphicx}
\usepackage{url}
\usepackage{theorem}
{\theorembodyfont{\rmfamily}
\newtheorem{definition}{\bfseries Definition}}
\newtheorem{proposition}{Proposition}
\newtheorem{theorem}{Theorem}
\newenvironment{proof}{{\bf \emph{Proof.} }}{\hfill $\Box$ \\} 
\begin{document}
\frenchspacing
\tolerance=5000
\title{Neutrosophic Graph Rewriting as a Computational Model}
\author{Apostolos Syropoulos\\ Xanthi, Greece\\ \texttt{asyropoulos@yahoo.com} \and
         Basil Papadopoulos\\ Department of Civil Engineering\\ Democritus University of Thrace\\ Greece\\ \texttt{papadob@civil.duth.gr}
\and     Binod Kumar Sharma\\ Department of Computer science and Engineering\\ Birla Institute of Technology\\ 
         Extension Center, Lalpur, Ranchi \\ Jharkhand\\ India\\ \texttt{bksharma@bitmesra.ac.in}}
\maketitle
\begin{abstract} 
Graphs, which are made up of points connected by lines, are used to represent relationships between objects. They are useful for solving problems in computer science, logistics, biology, and other fields, including social networks and the internet. Neutrosophic graphs are graphs whose points and lines are characterized by membership, indeterminacy, and non-membership degrees. In the simplest case, only lines are characterized by these degrees. The process by which we transform one graph to another one is called graph rewriting. It has been shown that graph rewriting can be described mathematically by pushouts, which are category-theoretic structures. It has also shown that fuzzy graphs have this property, and here we prove that neutrosophic graph rewriting can be described by pushouts in a certain category. By providing a fuzzy version of category theory, it is possible to define fuzzy categorical structures and thus to define neutrosophic graph rewriting in a more suitable manner. Since graph rewriting is actually a model of computation, neutrosophic graphs are a neutrosophic model of computation. A hypergraph is a generalization of a conventional graph. It allows a single ``hyperline'' to connect many points, which is different from regular graph lines that only connect two points. Neutrosophic hypergraphs can be used to reformulate Wolfram’s models, which are supposed to be a mathematical model of reality.
\end{abstract}
\section{Introduction}
A graph is a mathematical structure used to visually show the relationship between different items or quantities.
It is a fundamental concept in both mathematics and computer science. 
\begin{definition}
A graph, often denoted as $G=(V,E)$, consists of two main parts:
\begin{itemize}
\item Vertices (or Nodes): A set of objects, typically represented as points or circles. The set of all vertices is denoted by $V$.
\item Edges (or Links): Connections or relationships between pairs of vertices, represented as lines or curves. The set of all 
      edges is denoted by $E$.
\end{itemize}
\end{definition}
In computer science, graph rewriting, also known as graph transformation, is a method for methodically changing a graph by 
applying rewrite rules  that substitute a specific subgraph pattern with a new, different graph structure. Graph rewriting has found many
applications, for example:
\begin{itemize}
\item Graph Neural Networks (GNNs) are neural networks that operate on graphs instead of tabular data. In particular, GNNs  are essentially 
      graph transformations~\cite{machowczyk2023}.
\item Origami, the Japanese art of paper folding, has been formalized as graph rewriting~\cite{ida2008,ida2010}.
\item Compilers can optimize code by manipulating graph-like structures that represent intermediate code, which enables techniques 
      such as loop invariant computation and common subexpression elimination, through the use of graph rewriting~\cite{nagl1981}.
\item It has been shown that category theory offers specifications that can be efficiently implemented through imperative algorithms, and 
      this has been applied to the field of graph transformation. It was demonstrated through the use of examples how this software development 
      paradigm facilitates the rapid composition of efficient and accurate code~\cite{brown2022}.
\item It has been demonstrated how graph rewriting can been seen as a computational model~\cite{mcburney1991} and in the second section of
      this chapter we demonstrate this idea.
\end{itemize}
Pushouts are category theoretic constructs that are used to formalize the changing of a graph by applying rewrite rules---they precisely 
define how graphs are ``glued'' together during rewriting. A graph rewriting rule is typically given by a set of morphisms:
\begin{center}
\includegraphics[scale=1]{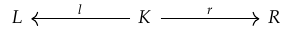}
\end{center}
where $L$ is the left-hand side (pattern to be replaced), $R$ is the right-hand side (replacement), $K$ is the part preserved during rewriting, $l$ 
and $r$ graph morphisms. Thus, $K$ is what is common between $L$ and $R$. In category theory, when given a set of arrows $f:A\rightarrow B$ and
$g:A\rightarrow C$ of some category $\mathcal{C}$, a pushout is a triple $(D, \alpha, \beta)$ with $\beta\circ g = \alpha\circ f$ that satisfies
the following universal property: for every triple $(E,\alpha',\beta')$ with $\beta'\circ g=\alpha'\circ f$, there exists a unique arrow 
$\theta:D\rightarrow E$ making the following diagram commute
\begin{center}
\includegraphics[scale=1]{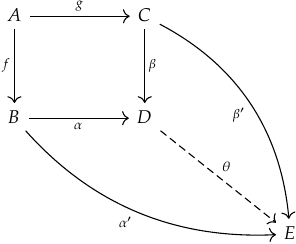}
\end{center}
Now, it should be obvious why pushouts are actually  graph rewriting rules. More generally, we can use two pushouts to define how to delete 
and then how to add subgraphs. From this step, we can define {\em pushout automata} by pushout composition of graphs. For a detailed description
of this technique, interested readers should consult the tutorial by Ehrig~\cite{ehrig1987}.

A {\em hypergraph} is a generalized graph where an edge (usually called {\em hyperedge}) can connect any number of vertices--two or more, or even 
just one. A hypergraph $H$ is a pair $H=(V,E)$, where $V$ is a set of vertices, and $E$ is a set of non-empty subsets of $V$, called hyperedges.
Steven Wolfram~\cite{wolfram2020} proposed some models that are some sort of graph rewriting that can be used to understand known fundamental 
features of physics. In general, the structures on which these models operate consist of collections of relations between identical, labeled,  
discrete elements. Wolfram proposed that these structures can be represented by hypergraphs. The elements are the nodes
of the hypergraph and the relations are the (directed) hyperedges that connect these elements. For example, the following graph
\begin{center}
\includegraphics[scale=1]{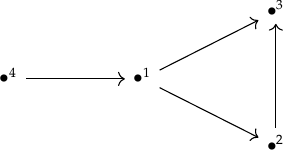}
\end{center}
corresponds to the following set of relations
\begin{displaymath}
\bigl\{ \{1, 2\}, \{1, 3\}, \{2, 3\}, \{4, 1\}\bigr\}.
\end{displaymath}
The order in which these relations are expressed is irrelevant; however, the order in which elements appear within each relation is considered 
significant (and is evident in the orientations of the hyperedges in the hypergraph). Wolfram's models are rules for rewriting rewriting collections 
of relations. A very simple rule is the following one
\begin{equation}\label{rule::1}
\{\{x, y\}\} \rightarrow \{\{x, y\}, \{y, z\}\},
\end{equation}
where $x$, $y$, and $z$ stand for any elements. Note that the following rule
\begin{equation}\label{rule::2}
\{\{x, y\}\} \rightarrow \{\{z, y\}, \{y, x\}\}
\end{equation}
is different from the previous one since order within each relation matters. Suppose we have the set of relations $\{\{1,2\}\}$ and
we apply rule~\ref{rule::2}, then we get the set of relations $\{\{7,2\},\{2,1\}\}$, where $7$ stands for a new node whose label
is free (i.e., not assigned to some other node). Now, let us apply the same rule~\ref{rule::2} to the new set of relations. The 
result will be the set $\{ \{8, 2\}, \{2, 7\}, \{9, 1\}, \{1, 2\}\}$. Ther graph that follows is a depiction of this resulting relation.
\begin{center}
\includegraphics[scale=1]{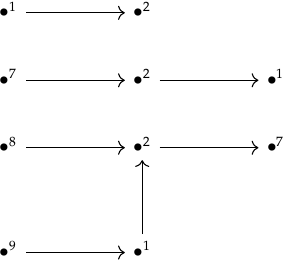}
\end{center}

Neutrosophic sets (sometimes called neutrosophic logic or neutrosophy) are a generalization of fuzzy sets and intuitionistic fuzzy sets, 
introduced by Florentin Smarandache~\cite{smarandache2006} in the late 1990s. Neutrosophic sets have used in computer science for the analysis 
of medical images~\cite{koundal2019}. In a neutrosophic set, each element $x$ is characterized by a triple of functions:
\begin{displaymath}
A(x)=\bigl(T(x),I(x),F(x)\bigr),
\end{displaymath}
where $T(x)$ is the degree of truth (membership), $I(x)$ is the degree of indeterminacy (uncertainty, unknown, inconsistent info), and
$F(x)$ is the degree of falsehood (non-membership). Since we are going to study graph rewriting of neutrosophic graphs we need to
provide a definition of these structures:
\begin{definition}
A {\em neutrosophic graph}~\cite{panda2025} is an ordinary graph (i.e., a set of vertices and a set of edges connecting
these vertices) accompanied by the following functions
\begin{enumerate}
\item a function  $T: V \rightarrow [0, 1]$ that represents truth;
\item a function $I: V \rightarrow [0, 1]$ that represents indeterminacy; and
\item a function $F: V \rightarrow [0, 1] $that represents falsity;
\end{enumerate}
such that $0\le T(v) + I(v) + F(v) \le 3$ for all $v\in V$. Clearly, these functions assign membership 
values to each vertex $v\in V$, representing the truth,
indeterminacy, and falsity degrees, respectively, of the proposition ``$v$ is a member of the graph $G$.''
Similarly, the membership values of an edge $(u, v)\in E$ are defined by the following three functions
\begin{enumerate}
\item a membership function $T: E \rightarrow [0, 1]$ that represents truth;
\item a membership function $I: E \rightarrow [0, 1]$ that represents indeterminacy; 
\item a membership function $F: E \rightarrow [0, 1]$ that represents falsity; 
\end{enumerate}
such that $0\le T(u, v) + I(u, v) + F(u, v) \le 3$ for all $(u, v) \in E$.
These  functions represent  the  truth,  indeterminacy,  and falsity degrees,  respectively,  of the
proposition  ``the edge $(u, v)$  is a member of the graph $G$.'' 
\end{definition}

{\bfseries\itshape Plan of the chapter}\hskip1em First we demonstrate how we can actually compute using graph rewriting, then we
present neutrosophic graph rewriting and define neutrosophic double pushout graph rewriting and neutrosophic hypergraph rewriting {\`a} la Wolfram.
\section{Graph Rewriting as a Model of Computation}
In this section we will describe a simple method to compute the successor of (positive) binary number, which is represented by a graph,
using graph rewriting. A binary number will be represented by a chain of nodes representing bits and the least significant bit  
will be at the end. There are two special nodes one at the beginning and one at the end of the graph. Also, the leftmost node that
containing a digit is the least significant bit. For example, the number $101$ ($5$ in decimal) might be represented as:
\begin{center}
\includegraphics[scale=1]{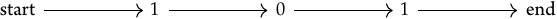}
\end{center}
The rules to compute the successor are:
\begin{enumerate}
\item If the current bit is 0, change it to 1 and stop (no carry needed).
\item If the current bit is 1, change it to 0 and propagate carry to the next bit.
\item If carry goes past the last bit, add a new 1 node.
\end{enumerate}
We can express these as rewrite patterns:
\begin{description}
\item[Rule 1] Increment a 0 bit
\begin{displaymath}
[0] \rightarrow \text{rest}     \Rightarrow    [1] \rightarrow \text{rest,}
\end{displaymath}
when no carry needed.
\item[Rule 2] Increment a 1 bit
\begin{displaymath}
[1] \rightarrow \text{rest}     \Rightarrow    [0] \rightarrow \text{(carry-to-rest),}
\end{displaymath}
propagate carry if necessary.
\item[Rule 3] Carry past the end
\begin{displaymath}
\text{(end with carry)}  \Rightarrow  [1] \rightarrow \text{end.}
\end{displaymath}
\end{description}
Let us apply these rules to number $110$  (binary 6) represented as:
\begin{center}

\includegraphics[scale=1]{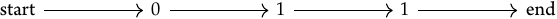}
\end{center}
In this case, we simply apply \textbf{Rule 1} and so $[0]$ becomes $[1]$ and we stop. Thus, we get 
\begin{center}
\includegraphics[scale=1]{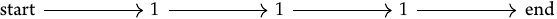}
\end{center}
which is the number $7$ in binary.

The previous example showed that we can actually compute with graph rewriting. However, this example
just verified that we can compute but it is not a proof that we can compute anything that is computable.
Max Dauchet~\cite{dauchet1992} has proved that it is possible to simulate Turing machines with graph rewriting.
Thus, graph rewriting is Turing-complete, which means that we can compute with it anything we can compute with
a Turing machine. 
\section{Neutrosophic Graph Rewriting}
Since neutrosophic graphs are a generalization of ordinary graphs, it makes sense to perform graph rewriting to such graphs.
However, authors of other works exploring graph rewriting under vagueness (e.g., see~\cite{mori1995}), consider only
fuzzy graphs whose vertices form a crisp set and not a fuzzy one. First, we will explore graph rewriting of a restricted 
version of neutrosophic graphs whose vertices do not have truth, falsity, and indeterminacy degrees. Then, 
we will briefly explain how we can categorically handle neutrosophic graphs whose set of vertices and edges are also neutrosophic sets.

In general, a fuzzy graph can be represented by a fuzzy relation~\cite{mori1995}. Thus, we can adopt this idea and we will
view neutrosophic graphs as binary neutrosophic relations:
\begin{definition}
A binary neutrosophic relation from a (crisp) set $X$ to a (crisp) $Y$ is a neutrosophic subset in $X\times Y$, that is, 
a structure $R$ given by
\begin{displaymath}
R = \bigl\{ \langle(x, y), \mu_{R}(x, y), \sigma_{R}(x, y), \nu_{R}(x, y)\rangle \mathrel{\bigm|} 
(x, y) \in X \times Y \bigr\},
\end{displaymath}
where $\mu_{R}: X\times Y \rightarrow [0, 1]$, $\sigma_{R}: X\times Y\rightarrow [0, 1]$, 
and $\nu_{R}: X\times Y \rightarrow [0, 1]$.
For any $(x, y)\in X\times Y$, the value $\mu_{R}(x, y)$ is called the degree of a membership of 
$(x, y)$ in $R$, $\sigma_{R}(x, y)$ is called the degree of indeterminacy of $(x, y)$ in $R$, 
and $\nu_{R}(x, y)$ is called the degree of non-membership of $(x, y)$ in $R$.
\end{definition} 
Using this definition we define a neutrosophic graph $G=(V,E)$ as a pair $(V,R)$ where $V$ is its set of vertices and $R$ is a 
binary neutrosophic relation from $V$ to $V$. Mizoguchi and Kawahara~\cite{mizoguchi1995} defined a {\em relational calculus} that
has been redefined for fuzzy relations by Mori and Kawahara~\cite{mori1995} and which we redefine for neutrosophic relations.
\begin{definition}
The (crisp) set of all neutrosophic relations from $A$ to $B$ will be denoted by $\mathbf{Rel}_{n}(A,B)$. This set contains
the relations $\mathbf{0}_{A,B}(a,b)=0$ and $\mathbf{U}_{A,B}(a,b)=1$, for all $(a,b)\in A\times B$. 
\end{definition}
Assume that $S,R\in\mathbf{Rel}_{n}(A,B)$, then $S\sqsubseteq R$ if and only if $\mu_{S}(a,b)\leq \mu_{R}(a,b)$, 
$\sigma_{S}(a,b)\leq \sigma_{R}(a,b)$, and $\nu_{S}(a,b)\geq \nu_{R}(a,b)$ for all $(a,b)\in A\times B$.
\begin{proposition}\label{poset::1}
The set $\mathbf{Rel}_{n}(A,B)$ equipped with the relation $\sqsubseteq$ is a poset. 
\end{proposition}
\begin{proof}
The relation $\sqsubseteq$ is clearly reflexive, transitive, and  anti-symmetric.  Ergo, $\mathbf{Rel}_{n}(A,B)$ equipped with the relation $\sqsubseteq$
is a poset. 
\end{proof}

Given two relations $S,R\in\mathbf{Rel}_{n}(A,B)$, we define the operations $S\sqcap R$ and $S\sqcup R$ as follows:
\begin{align*}
\mu_{S\sqcap R}(a,b)&=\min\bigl\{ \mu_{S}(a,b),\mu_{R}(a,b)\bigr\}, & \mu_{S\sqcup R}(a,b)&=\max\bigl\{ \mu_{S}(a,b),\mu_{R}(a,b)\bigr\}\\
\sigma_{S\sqcap R}(a,b)&=\min\bigl\{ \sigma_{S}(a,b),\sigma_{R}(a,b)\bigr\}, & \sigma_{S\sqcup R}(a,b)&=\max\bigl\{ \sigma_{S}(a,b),\sigma_{R}(a,b)\bigr\}\\
\nu_{S\sqcap R}(a,b)&=\max\bigl\{ \nu_{S}(a,b),\nu_{R}(a,b)\bigr\}, & \nu_{S\sqcup R}(a,b)&=\min\bigl\{ \nu_{S}(a,b),\nu_{R}(a,b)\bigr\}\\
\end{align*}
When the operants of ``$\sqcap$'' and ``$\sqcup$'' are not finite, we have to replace the operator ``$\min$'' with the operator ``$\inf$'' 
and the operator ``$\max$'' with the operator ``$\sup$,'' respectively. 
\begin{proposition}\label{poset::2}
The quadruple $\bigl(\mathbf{Rel}_{n}(A,B),\sqsubseteq,\sqcap,\sqcup\bigr)$ is a lattice with a zero element.
\end{proposition}
\begin{proof}
Any poset is a lattice if and only if all operations $S_1\sqcap S_2$ and $S_1\sqcup S_2$ are defined. Obviously for any $S_1, S_2\in\mathbf{Rel}_{n}(A,B)$,
these operations are defined, therefore  $\bigl(\mathbf{Rel}_{n}(A,B),\sqsubseteq,\sqcap,\sqcup\bigr)$ is a lattice.
\end{proof}
Now, let us define the relation $S\Rightarrow R$:
\begin{align*}
\mu_{S\Rightarrow R}(a,b) &= \left.
  \begin{cases}
    1,      & \text{if } \mu_{S}(a,b)\le \mu_{R}(a,b)\\
    \mu_{R}(a,b), & \text{if } \mu_{S}(a,b)> \mu_{R}(a,b)
  \end{cases}\right.\\
\sigma_{S\Rightarrow R}(a,b) &= \left.
  \begin{cases}
    1,      & \text{if } \sigma_{S}(a,b)\le \sigma_{R}(a,b)\\
    \sigma_{R}(a,b), & \text{if } \sigma_{S}(a,b)> \sigma_{R}(a,b)
  \end{cases}\right.\\
\nu_{S\Rightarrow R}(a,b) &= \left.
  \begin{cases}
    0,      & \text{if } \nu_{S}(a,b)\ge \nu_{R}(a,b)\\
   \nu_{S}(a,b), & \text{if } \nu_{S}(a,b)< \nu_{R}(a,b)
  \end{cases}\right.\\
\end{align*}
It is clear that $S\Rightarrow R\in\mathbf{Rel}_{n}(A,B)$ and the operator ``$\Rightarrow$'' is used to define the relative pseudo-complement 
in the lattice $\mathbf{Rel}_{n}(A,B)$.
\begin{proposition}
Suppose that $R, S\in\mathbf{Rel}_{n}(A,B)$. Then, for any $Q\in\mathbf{Rel}_{n}(A,B)$, it it true that
\begin{displaymath}
R\sqcap Q\sqsubseteq S\quad\text{if and only if}\quad Q\sqsubseteq R\Rightarrow S.
\end{displaymath}
\end{proposition}
\begin{proof}
Suppose that $Q\sqsubseteq R\Rightarrow S$. Then, this means that
\begin{align*}
\mu_{Q}(a,b) &\le \mu_{R\Rightarrow S}(a,b)\\
\sigma_{Q}(a,b) &\le \sigma_{R\Rightarrow S}(a,b)\\
\nu_{Q}(a,b) &\ge \nu_{R\Rightarrow S}(a,b)\\
\end{align*}
for any pair $(a,b)$.
\begin{align*}
\mu_{R\sqcap(R\Rightarrow S)}(a,b) &= \left.
  \begin{cases}
    R(a,b),      & \text{if } \mu_{R}(a,b)\le \mu_{S}(a,b)\\
    \mu_{R\sqcap S}(a,b), & \text{if } \mu_{R}(a,b)> \mu_{S}(a,b)
  \end{cases}\right.\\
\sigma_{R\sqcap(R\Rightarrow S)}(a,b) &= \left.
  \begin{cases}
    R(a,b),      & \text{if } \sigma_{R}(a,b)\le \sigma_{S}(a,b)\\
    \sigma_{R\sqcap S}(a,b), & \text{if } \sigma_{R}(a,b)> \sigma_{S}(a,b)
  \end{cases}\right.\\
\nu_{R\sqcap(R\Rightarrow S)}(a,b) &= \left.
  \begin{cases}
    R(a,b),                  & \text{if } \sigma_{R}(a,b)\ge \sigma_{S}(a,b)\\
    \sigma_{R\sqcap S}(a,b), & \text{if } \sigma_{R}(a,b)< \sigma_{S}(a,b)
  \end{cases}\right.\\
\end{align*}
This means that $R\sqcap(R\Rightarrow S)=R\sqcap S$. We have assumed that 
$R\sqcap Q\sqsubseteq R\sqcap(R\Rightarrow S)=R\sqcap S \sqsubseteq S$. Conversely,
let us suppose that $R\sqcap Q\sqsubseteq S$. If 
\begin{align*}
\mu_{R}(a,b) &\le \mu_{Q}(a,b),\\
\sigma_{R}(a,b) &\le \sigma_{Q}(a,b),\\
\intertext{and}
\nu_{R}(a,b) &\ge \nu_{Q}(a,b),\\
\end{align*}
then 
\begin{align*}
\mu_{R}(a,b)    &= \mu_{R\sqcap Q}(a,b)   \le \mu_{S}(a,b),\\
\sigma_{R}(a,b) &= \sigma_{R\sqcap Q}(a,b) \le \sigma_{S}(a,b),\\
\intertext{and}
\nu_{R}(a,b)    &= \nu_{R\sqcap Q}(a,b)   \ge \nu_{S}(a,b).\\
\end{align*}
This means that
\begin{align*}
\mu_{R\Rightarrow S}(a,b)    &= 1,\\
\sigma_{R\Rightarrow S}(a,b) &=1,\\
\intertext{and}
\nu_{R\Rightarrow S}(a,b) &=0.\\
\end{align*} 
Therefore it holds that
\begin{align*}
\mu_{Q}(a,b)    &\le \mu_{R\Rightarrow S}(a,b),\\
\sigma_{Q}(a,b) &\le \sigma_{R\Rightarrow S}(a,b),\\
\intertext{and}
\nu_{Q}(a,b)    &\ge \nu_{R\Rightarrow S}(a,b).\\
\end{align*}
Otherwise, it should be true that 
\begin{align*}
\mu_{R}(a,b) &> \mu_{Q}(a,b),\\
\sigma_{R}(a,b) &> \sigma_{Q}(a,b),\\
\intertext{and}
\nu_{R}(a,b) &< \nu_{Q}(a,b),\\
\end{align*}
but we have assumed that 
\begin{align*}
\mu_{Q}(a,b)    &= \mu_{R\sqcap Q}(a,b)   \le \mu_{S}(a,b) \le \mu_{R\Rightarrow S}(a,b),\\
\sigma_{Q}(a,b) &= \sigma_{R\sqcap Q}(a,b)   \le \sigma_{S}(a,b) \le \sigma_{R\Rightarrow S}(a,b),\\
\intertext{and}
\nu_{Q}(a,b)    &= \nu_{R\sqcap Q}(a,b)   \ge \nu_{S}(a,b) \ge \nu_{R\Rightarrow S}(a,b).\\
\end{align*}
Hence, $Q\sqsubseteq(R\Rightarrow S)$.
\end{proof}
Usually, we denote the relatively pseudo-complement $T\Rightarrow\mathbf{0}$ by $\neg T$. The following result can be
proved using the previous results.
\begin{theorem}
For any sets $A$ and $B$, the quadruple $\bigl(\mathbf{Rel}_{n}(A,B),\sqsubseteq,\sqcap,\sqcup\bigr)$ is a Heyting algebra.
\end{theorem}
\begin{definition}
A neutrosophic relation $R\in \mathbf{Rel}_{n}(A,B)$ is {\em regular} if and only if
$\mu_R(a,b)=0$, $\sigma_R(a,b)=0$, and  $\nu_R(a,b)=1$ or $\mu_R(a,b)=1$, $\sigma_R(a,b)=1$, and $\nu_R(a,b)=0$.
\end{definition}
\begin{definition}
Assume that $R, S\in \mathbf{Rel}_{n}(A,B)$ and $R\sqsubseteq S$. Then, $R$ is {\em regular with respect to} $S$ 
if and only if $\mu_{R}(a,b)\neq 0$ implies $\mu_{R}(a,b)=\mu_{S}(a,b)$, $\sigma_{R}(a,b)\neq 0$ implies $\sigma_{R}(a,b)=\sigma_{S}(a,b)$,
and $\nu_{R}(a,b)\neq 1$ implies $\nu_{R}(a,b)=\nu_{S}(a,b)$.
\end{definition}
\begin{definition}
Assume that $0<\epsilon_1<1$, $0<\epsilon_2<1$, and $0<\epsilon_3<1$. Then, a neutrosophic relation $R$ is $\varepsilon$-regular 
with respect to $T$ if for all $a$ and $b$
\begin{align*}
\mu_{R}(a,b)\neq 0 &\Rightarrow \mu_{T}(a,b)\neq0\;\text{and}\;|\mu_{R}(a,b)-\mu_{T}(a,b)|\le\epsilon_1,\\
\sigma_{R}(a,b)\neq 0 &\Rightarrow \sigma_{T}(a,b)\neq0\;\text{and}\;|\mu_{R}(a,b)-\mu_{T}(a,b)|\le\epsilon_2,\;\text{and}\\
\nu_{R}(a,b)\neq 0 &\Rightarrow \nu_{T}(a,b)\neq0\;\text{and}\;|\mu_{R}(a,b)-\mu_{T}(a,b)|\le 1-\epsilon_3.\\
\end{align*}
\end{definition}
\begin{definition}
Given a neutrosophic relation $R:A\rightarrow B$, the domain $d(R):A\rightarrow A$ of $R$ is a relation defined by
\begin{displaymath}
d(R) = R\circ R^{-1}\circ\mathrm{id}_{A}, 
\end{displaymath}
where $R^{-1}$ is the {\em inverse} relation of $R$, $R\circ S$ is the {\em composite} relation of $R$ and $S$, and $\mathrm{id}_{A}$ is the 
{\em identity} relation of $A$ (the diagonal set of A).
\end{definition}
\begin{definition}
Given two neutrosophic graphs $(V_1,R)$ and $(V_2,S)$, a {\em partial} graph morphism from $(V_1,R)$ to $(V_2,S)$ is a 
regular-partial\footnote{Recall that any function is a relation and a (crisp) relation is a neutrosophic relation\ldots} 
function $f:V_1\rightarrow V_2$ satisfying $d(f)\circ R \circ f\sqsubseteq f\circ S$.
\end{definition} 
It is not difficult to show that all neutrosophic graphs and the partial graph morphisms between them form a category,
which will be denoted by $\mathbf{Pfn}(\mathbf{NGraph})$. Using the same arguments presented in~\cite{mizoguchi1995}, we
can prove the following:
\begin{theorem}
The category $\mathbf{Pfn}(\mathbf{NGraph})$ has pushouts.
\end{theorem}
In order to apply a rewriting rule to a graph we need a mechanism that will check whether the rule can be applied to a
subgraph. This mechanism is called matching. Let us now formalize rewriting rules and matching:
\begin{definition}
A rewriting rule is a triple $p=\bigl((V_1,R), (V_2,S), f:V_1\rightarrow V_2\bigr)$, where $f$ is a partial function.
\end{definition}
\begin{definition}
An arrow $g$ from $(V,R)$ to $(G,\Xi)$ is a {\em matching} from $(V,R)$ to $(G,\Xi)$ if $g^{-1}\circ R\circ g$ is regular with
respect to $\Xi$.
\end{definition}

In the crisp case, we can construct the pushout from the rewriting rule and matching. For example, the following pushout square
exists in category $\mathbf{Pfn}$ of sets and partial functions between them. 
\begin{center}
\includegraphics[scale=1]{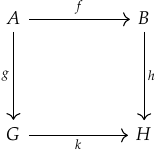}
\end{center}
The graph $(G,\Xi)$ is transformed into $(H,\Lambda)$ by applying a rewriting rule $p$ together with a matching $g$ if 
$\Lambda=k^{-1}\circ T\circ k\sqcup h^{-1}\circ(\Xi-g^{-1}\circ R\circ g)\circ k$. It is customary to denote the previous
rewriting square by
$(G,\Xi)\mathrel{\Longrightarrow_{p/g}}(H,\Lambda)$. The application of a rewriting rule can be viewed as a rewriting square:
\begin{center}
\includegraphics[scale=1]{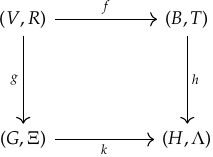}
\end{center}
Mori and Kawahara~\cite{mori1995} have shown that the rewriting square is a pushout square in $\mathbf{Pfn}(\mathbf{FGraph})$,
the category of fuzzy graphs and partial morphisms between them. Thus, we can prove the following theorem for category 
$\mathbf{Pfn}(\mathbf{NGraph})$.
\begin{theorem}
Suppose that $p$ is a rewriting rule $\bigl((V,R),(W,S), f:V\rightarrow W\bigr)$ and $g:(V,R)\rightarrow(W,S)$ is a matching.
If $f:(V,R)\rightarrow(W,S)$ is a partial morphism, then the rewriting square $(G,\Xi)\mathrel{\Longrightarrow_{p/g}}(H,\Lambda)$
is a pushout in $\mathbf{Pfn}(\mathbf{NGraph})$.
\end{theorem}
The next thing we need to define is the vague matching condition.
\begin{definition}
Assume that $\epsilon_1, \epsilon_2, \epsilon_3\in[0,1]$. Then, $g:V\rightarrow G$ is an $\varepsilon$-matching from
$(V,R)$ to $(G,\Xi)$ if and only if the relation $g^{-1}\circ R\circ g$ is $\varepsilon$-regular with respect to $\Xi$.
\end{definition}
Two neutrosophic graphs $(V,R)$ and $(W,S)$ are $\varepsilon$-equal if and only if $R$ is $\varepsilon$-regular 
with respect to $S$ and vice versa.

An alternative approach to the categorical representation of fuzzy graphs has been presented in~\cite{rosset2023}.
Here we adopt the theory developed in~\cite{rosset2023} to the case of neutrosophic graphs. 
\begin{definition}
Assume that $\mathcal{L}=[0,1]^{3}$ and that $(\mathcal{L},\le)$ is a poset. Then, a neutrosophic fuzzy set is
a pair $(A,\alpha)$ consisting of a set $A$ and a membership function 
$\alpha:A\rightarrow\mathcal{L}$. A morphism $f:(A,\alpha)\rightarrow(B,\beta)$ of fuzzy sets
is a function $f:A\rightarrow B$ such that $\alpha\le \beta\circ f$, that is, $\alpha(a)\le (\beta\circ f)(a)$
for all $a\in A$. They form the category $\mathbf{NFuzzySet}(\mathcal{L})$.
Membership functions are represented in gray. The fact that $\alpha\le \beta\circ f$
is expressed by the following diagram, which is called $\le$-commuting.
\begin{center}
\includegraphics[scale=1]{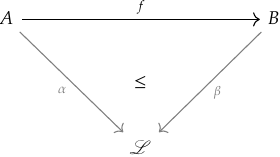}
\end{center}
\end{definition}
\begin{definition}
A neutrosphic graph $(A, \alpha)$ consists of a graph $A$ and
two membership functions $\alpha_V : A(V)\rightarrow [0,1]^{3}$, and 
$\alpha_E:A(E)\rightarrow [0,1]^{3}$, where $[0,1]^{3}$ is a posets. 
A neutrosophic graph homomorphism $f : (A,\alpha)\rightarrow (B, \beta)$ 
is a graph homomorphism $f : A\rightarrow B$ such that $\alpha_V\le\beta_V\circ f_V$ and 
$\alpha_E \le \beta_E\circ f_E$. All neutrosphic graphs and the neutrosophic graph 
homomorphism between them form the category $\mathbf{NGraph}$.
\end{definition}
\begin{definition}
A presheaf on a category $\mathbf{I}$ is a functor $F:\mathbf{I}^{\mathrm{op}}\rightarrow\mathbf{Set}$. 
An arrow of presheaves is a natural transformation between functors. The category of presheaves on 
$\mathbf{I}$ is denoted by $\mathbf{Set}^{\mathbf{I}^{\mathrm{op}}}$.
\end{definition}
Suppose that  $\mathsf{incl}:\mathbf{Set}\rightarrow\mathbf{Poset}$ the inclusion functor that gives the
coarsest possible partial order, that is, the reflexive relation.
\begin{definition}
Given a category $\mathbf{I}$ and a functor $L:\mathbf{I}^{\mathrm{op}}\rightarrow\mathbf{Poset}$, 
a (poset) $L$-fuzzy presheaf is a pair $(A, \alpha)$ consisting of a presheaf 
$A:\mathbf{I}^{\mathrm{op}}\rightarrow\mathbf{Set}$ and a lax natural transformation 
$\alpha: \mathsf{incl}\cdot A \Rightarrow L$, that is, a natural transformation whose naturality 
squares $\le$-commute, as shown in the diagram below
right for each $\iota:j\rightarrow i$. A morphism of $L$-fuzzy presheaves
$f: (A,\alpha)\rightarrow (B,\beta)$ is a natural transformation $f: A\rightarrow B$
such that $\alpha \le \beta\circ f$, that is, $\alpha_i\le\beta_i\circ f_i$ for all $i\in\mathbf{I}$. 
The result is the category $\mathbf{FuzzyPresheaf}(\mathbf{I}, L)$.
\begin{center}
\includegraphics[scale=1]{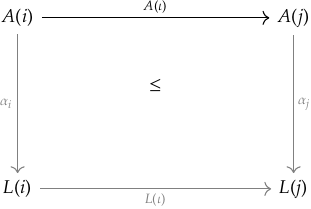}
\end{center}
\end{definition}
\begin{theorem}
Suppose that $\mathbf{I}$ is a small category and a functor $L:\mathbf{I}^{\mathrm{op}}\rightarrow\mathbf{CompHeytAlg}$.
Then, $\mathbf{FuzzyPresheaf}(\mathbf{I}, L)$ is a quasitopos.
\end{theorem}
Since a quasitopos is a category for which every finite limit and finite colimit exists, this means that pushouts exists.
Therefore, graph rewriting rules can be represented as pushouts. In order to avoid lax natural transformation  and 
other such constructs, we can adopt a definition of fuzzy categories such as those defined in~\cite{syropoulos2020}. 
\section{Neutrosophic Hypergraph Rewriting}
Neutrosophic hypergraphs of the kind we are interested in are defined as follows~\cite{akram2018}
\begin{definition}
Assume that $V = \{v_1, v_2,\ldots, v_n\}$ be a finite set and $E = \{E_1, E_2,\ldots, E_m\}$ be a
finite family of non-trivial single-valued neutrosophic subsets of $V$ such that
\begin{displaymath}
V =\bigcup_{i} \mathalpha{\mathrm{supp}}(E_i), i =1, 2,\ldots, m,
\end{displaymath}
where the edges $E_i$ are single-valued neutrosophic subsets of $V$, 
\begin{displaymath}
E_i = \bigl\{\bigl(v_j , T_{E_i}(v_j), I_{E_i}(v_j),F_{E_i}(v_j)\bigr)\mathrel{\bigm|}j=1,2,\ldots,k\bigr\}, 
\end{displaymath}
$E_i\neq\emptyset$, for $i = 1, 2,\ldots, m$. Then, the pair $H = (V, E)$ is a single-valued
neutrosophic hypergraph on $V$, $E$ is the family of single-valued neutrosophic hyperedges of
$H$ and $V$ is the crisp vertex set of $H$.
\end{definition}

We can use the neutrosophic sets $E_i$ to form the cartesian product $\mathcal{D}=D_1\times D_2\times\cdots\times D_m$.
Then, we can define a neutrosophic hypergraph by $(V,R)$, where $R\subseteq\mathcal{D}$. Thus, we have managed to transform
a neutrosophic hypergraph to a neutrosophic relation. This means that we can generalize the relational calculus of the
previous section for neutrosophic hypergraphs and so we can define hypergraph rewriting as a pushout in some category. 
In addition, it is possible to reformulate Wolfram's models as rules for rewriting rewriting collections 
of neutrosophic relations. Of course this is useful once we accept that vagueness is a fundamental property of this world.

\section{Conclusions}
We have explained why graph rewriting is a model of computation and then we defined categorically neutrosophic graph rewriting. In
addition, we showed how an alternative approach to categorical neutrosophic graph rewriting can be used for ``full'' neutrosophic graphs.
Moreover, we briefly examined neutrosophic hypergraph rewriting and its possible applications.  We believe this is an exciting field of
study and research as there are many things to be examined and studied.

\end{document}